\documentclass[journal]{IEEEtran}
\IEEEoverridecommandlockouts
\usepackage{cite}
\usepackage{hyperref}
\usepackage{amsmath,amssymb,amsfonts}
\usepackage{amsmath}
\usepackage{amsthm}

\usepackage{graphicx}
\usepackage{textcomp}
\usepackage{xcolor}
\usepackage{graphicx}
\usepackage{float} 
\usepackage{subfigure}
\usepackage{amsmath}
\usepackage{amsfonts,amssymb}
\usepackage{mathrsfs}
\usepackage{mathtools}
\usepackage{algorithm}
\usepackage{algorithmicx}
\usepackage{algpseudocode}
\usepackage{bm}
\bibliography{IEEEabrv}
\def\BibTeX{{\rm B\kern-.05em{\sc i\kern-.025em b}\kern-.08em
    T\kern-.1667em\lower.7ex\hbox{E}\kern-.125emX}}
\newtheorem{lemma}{Lemma}
\begin{document}

\title{Ergodic Mutual Information for Generalized Fadings\\
}


\author{Chongjun~Ouyang,~\IEEEmembership{Student Member,~IEEE}, 
        Sheng~Wu,
	   Zeliang~Ou, 
	   Pei~Yang,
	   Lu~Zhang,
        Xin~Zhang,
and Hongwen~Yang,~\IEEEmembership{Member,~IEEE}    
\thanks{C. Ouyang, S. Wu, Z. Ou, P. Yang, L. Zhang, X. Zhang and H. Yang are with the School of Information and Communication Engineering, Beijing University of Posts and Telecommunications, Beijing 100876, China. (E-mail: \{DragonAim, thuraya, ouzeliang, yp, zhangl\_96, zhangxin, yanghong\}@bupt.edu.cn)}
}


\maketitle

\begin{abstract}

In this paper, novel expressions are derived to evaluate the ergodic mutual information (EMI) under BPSK modulation of single-input single-output (SISO) systems operating in generalized fading channels, including $\eta$-$\mu$ fading and $\kappa$-$\mu$ fading. To verify our derivation, we first investigate the specific fading types, namely Rayleigh, Nakagami-$m$ and Rician, and then turn to generalized fading scenarios. It is shown that all the expressions concluded from the generalized cases can be specialized into those derived from the specific ones. Different from the conventional results of the EMI, our developed expressions contain only the simplest numerical calculations, without any Meijer's G-functions which must be implemented in the particular computing software. Additionally, it should be noted that our work provides a complete analysis of the EMI in wireless channel under discrete inputs.

\end{abstract}


\section{Introduction}
Ergodic mutual information (EMI) is an important performance measurement metric in wireless multi-path channels, indicating the maximal achievable transmission rate. Many researches on the EMI in different fading scenarios have been presented in the past decades, including Rayleigh  \cite{c1}, Nakagami-$m$  \cite{c2}, Rician  \cite{c3} fading channels, and  generalized fading channels, namely $\eta$-$\mu$ \cite{c4} and $\kappa$-$\mu$  \cite{c5}. However, all of these works only considered the simplest scenario when the input signals followed Gaussian distribution. In fact, a practical implementation is the case where the channel inputs are drawn from finite alphabet, which is necessary to be studied. Although the works in \cite{c6} have derived the recursive expression for the EMI under BPSK in Nakagami-$m$ fading channels, the computation complexity is prohibitively high for large $m$.

In this paper, we derive accurate and neat expressions for the ergodic mutual information of single-input single-output (SISO) systems under BPSK operating both in specific fading channels, including Rayleigh, Nakagami-$m$ and Rician, and generalized fading channels including $\eta$-$\mu$ and $\kappa$-$\mu$. It is difficult to derive the closed-form exact formulas of EMI. Instead we develop some simple closed-form approximate expressions with high precision to estimate the maximal achievable transmission rate. Notably, our proposed expressions encompass no complicated structures, such as the commonly-used Meijer's G-functions in the conventional works \cite{c4,c5}. Besides, our derived results of the generalized fading channels can accord with those of specific fadings, which reveal the feasibility and validity of our developed formulas. To the best of our knowledge, this is the first time to propose a comprehensive analysis for the ergodic mutual information with finite-alphabet inputs under generalized multi-path fading types.


\section{Derivation of Mutual Information}
Denote $f\left(\cdot\right)$ as the probability density function (PDF) of the instantaneous Signal to Noise Ratio (SNR) per symbol. Then, the ergodic mutual information is defined as
\begin{equation}
\label{eq1}
C\triangleq\int_{0}^{+\infty}{\mathcal{I}}\left(\gamma\right)f\left(\gamma\right){\rm{d}}\gamma,
\end{equation}
where ${\mathcal{I}}\left(\cdot\right)$ represents the instantaneous input-output mutual information as a function of the SNR $\gamma$.

\label{section2}
\subsection{AWGN}
Assume that the transmitted data stream are independent and identically distributed (i.i.d.) zero-mean binary symbols with equal probabilities, and then the input-output mutual information in terms of the SNR $\gamma$ under BPSK modulation over additive white Gaussian noise (AWGN) channels can be formulated as \cite{b1}
\begin{equation}
\label{eq2}
{\mathcal{I}}\left(\gamma\right)=1-\int_{-\infty}^{+\infty}{\frac{1}{\sqrt{2\pi}}{\rm{e}}^{-\frac{u^2}{2}}\log_2\left(1+{\rm{e}}^{-2\sqrt{\gamma}u-2\gamma}\right){\rm{d}}u}.
\end{equation}
An approximate formula of ${\mathcal{I}}\left(\gamma\right)$, with compact form and high accuracy, is expressed as \cite{b4}
\begin{equation}
\label{eq3}
{\mathcal{I}}\left(\gamma\right)\approx1-{\rm{e}}^{-\vartheta\gamma},
\end{equation}
where $\vartheta=0.6507$. The precision of Equ. \eqref{eq3} will be examined later.  

\subsection{Rayleigh Fading Scenario}
The PDF of received SNR in Rayleigh fading channels can be written as 
\begin{equation}
\label{eq4}
f\left(\gamma\right)=\frac{1}{\bar\gamma}{\rm{e}}^{-\frac{\gamma}{\bar\gamma}},
\end{equation}
where $\bar\gamma$ denotes the average SNR.
\begin{lemma}
The ergodic mutual information in Rayleigh fading channels is approximated as
\begin{equation}
\label{eq5}
{\hat{C}}=1-\frac{1}{1+\vartheta\bar\gamma}.
\end{equation} 
\end{lemma}
\begin{proof}
On the basis of Equ. \eqref{eq3} and Equ. \eqref{eq4}, the ergodic mutual information is derived as follows: 
\begin{equation}
\label{eq6}
\begin{split}
{C}&=\int_{0}^{+\infty}{\mathcal{I}}\left(\gamma\right)f\left(\gamma\right){\rm{d}}\gamma\approx1-\int_{0}^{+\infty}{\rm{e}}^{-\vartheta\gamma}\frac{1}{\bar\gamma}{\rm{e}}^{-\frac{\gamma}{\bar\gamma}}{\rm{d}}\gamma\\
&=1-\frac{1}{1+\vartheta\bar\gamma}.
\end{split}
\end{equation} 
\end{proof}

\subsection{Nakagami-$m$ Fading Scenario}
The PDF of received SNR under Nakagami-$m$ fading is given by
\begin{equation}
\label{eq7}
f\left(\gamma\right)=\frac{m^m\gamma^{m-1}}{\Gamma\left(m\right)\bar\gamma^m}{\rm{e}}^{-\frac{m\gamma}{\bar\gamma}}.
\end{equation}
where $\Gamma\left(x\right)=\int_{0}^{+\infty}t^{x-1}{\rm{e}}^{-t}{\rm{d}}t$ denotes the Gamma function.
\begin{lemma}
The ergodic mutual information under Nakagami-$m$ fading is approximated as
\begin{equation}
\label{eq8}
{\hat{C}}=1-\left(\frac{m}{m+\vartheta\bar\gamma}\right)^m.
\end{equation}
\end{lemma}
\begin{proof}
On the basis of Equ. \eqref{eq3} and Equ. \eqref{eq7},
\begin{equation}
\label{eq9}
\begin{split}
C&\approx1-\int_{0}^{+\infty}{\rm{e}}^{-\vartheta\gamma}\frac{m^m\gamma^{m-1}}{\Gamma\left(m\right)\bar\gamma^m}{\rm{e}}^{-\frac{m\gamma}{\bar\gamma}}{\rm{d}}\gamma\\
&=1-\frac{m^m}{\Gamma\left(m\right)\bar\gamma^m}\int_{0}^{+\infty}{\rm{e}}^{-\frac{m\gamma}{\bar\gamma}-\vartheta\gamma}\gamma^{m-1}{\rm{d}}\gamma.
\end{split}
\end{equation}
By \cite[Equ. (3.326.2)]{b2}, the final result can be summarized as
\begin{equation}
\label{eq10}
C\approx1-\left(\frac{m}{m+\vartheta\bar\gamma}\right)^m.
\end{equation} 
\end{proof}

\subsection{Rician Fading Scenario}
The PDF of received SNR under Rician fading is given by
\begin{equation}
\label{eq11}
\begin{split}
&f\left(\gamma\right)\\
&=\frac{1+K}{\bar\gamma}\exp\left(-K-\frac{\left(1+K\right)\gamma}{\bar\gamma}\right)I_0\left(2\sqrt{\frac{K\left(K+1\right)}{\bar\gamma}\gamma}\right),
\end{split}
\end{equation}
where $I_0\left(x\right)=\sum_{n=0}^{+\infty}\frac{\left(x/2\right)^{2n}}{n!\Gamma\left(n+1\right)}$ is the 0-th order modified Bessel function of the first kind.
\begin{lemma}
The ergodic mutual information under Rician fading is approximated as
\begin{equation}
\label{eq12}
{\hat{C}}=1-\frac{1}{1+\frac{\vartheta{\bar\gamma}}{K+1}}\exp\left(\frac{K}{1+\frac{\vartheta{\bar\gamma}}{K+1}}-K\right).
\end{equation}
\end{lemma}
\begin{proof}
On the basis of Equ. \eqref{eq3} and Equ. \eqref{eq11},
\begin{equation}
\label{eq13}
\begin{split}
C\approx&1-\int_{0}^{+\infty}\frac{1+K}{\bar\gamma}\exp\left(-K-\frac{\left(1+K\right)\gamma}{\bar\gamma}-\vartheta\gamma\right)\\
&\times I_0\left(2\sqrt{\frac{K\left(K+1\right)}{\bar\gamma}\gamma}\right){\rm{d}}\gamma\\
=&1-\sum_{{n=0}}^{+\infty}\frac{\left(K+1\right)^{n+1}}{\bar\gamma^{n+1}{n!}}\frac{{\rm{e}}^{-K}K^n}{\Gamma\left(n+1\right)}\int\limits_{0}^{+\infty}\gamma^n{\rm{e}}^{-\left(\vartheta+\frac{1+K}{\bar\gamma}\right)\gamma}{\rm{d}}\gamma\\
\overset{\left(a\right)}{=}&1-\sum_{n=0}^{+\infty}\frac{\left(K+1\right)^{n+1}}{\bar\gamma^{n+1}n!}\frac{{\rm{e}}^{-K}K^n}{\Gamma\left(n+1\right)}\frac{\Gamma\left(n+1\right)}{\left(\vartheta+\frac{1+K}{\bar\gamma}\right)^{n+1}}\\
=&1-\frac{{\rm{e}}^{-K}}{1+\frac{\vartheta{\bar\gamma}}{K+1}}\sum_{n=0}^{+\infty}\frac{1}{n!}\left(\frac{K}{1+\frac{\vartheta{\bar\gamma}}{K+1}}\right)^n\\
\overset{\left(b\right)}{=}&1-\frac{1}{1+\frac{\vartheta{\bar\gamma}}{K+1}}\exp\left(\frac{K}{1+\frac{\vartheta{\bar\gamma}}{K+1}}-K\right),
\end{split}
\end{equation}
where the step ``(a)" is based on \cite[Equ. (3.326.2)]{b2} and the step ``(b)" is based on the Taylor series expansion, $\exp(x)=\sum_{n=0}^{+\infty}\frac{1}{n!}x^n$. 
\end{proof}

\subsection{$\eta$-$\mu$ Fading Scenario}
The PDF of received SNR under {$\eta$-$\mu$ fading is given by \cite{b5}
\begin{equation}
\label{eq14}
f\left(\gamma\right)=\frac{2\sqrt{\pi}\mu^{\mu+\frac{1}{2}}h^\mu}{\Gamma\left(\mu\right)H^{\mu-\frac{1}{2}}}\frac{\gamma^{\mu-\frac{1}{2}}}{\bar\gamma^{\mu+\frac{1}{2}}}\exp\left(-\frac{2\mu\gamma h}{\bar\gamma}\right)I_{\mu-\frac{1}{2}}\left(\frac{2\mu H\gamma}{\bar\gamma}\right),
\end{equation}
where $I_\nu\left(z\right)=\left(\frac{1}{2}z\right)^\nu\sum_{n=0}^{+\infty}\frac{\left(\frac{1}{4}z^2\right)^{n}}{n!\Gamma\left(\nu+n+1\right)}$ is the modified Bessel function of the first kind.
It is well known that the $\eta$-$\mu$ fading holds two formats depending on the correlation in the main cluster. Additionally, $h$ and $H$ are functions of $\eta$ and vary from one format to another. More specifically, in Format 1, $0<\eta<\infty$, $h=\left(2+\eta^{-1}+\eta\right)/4$ and $H=\left(\eta^{-1}-\eta\right)/4$, whereas, in Format 2, $-1<\eta<1$, $h = 1/\left(1-\eta^2\right)$ and $H = \eta/\left(1-\eta^2\right)$.
\begin{lemma}
The ergodic mutual information under $\eta$-$\mu$ fading is approximated as
\begin{equation}
\label{eq15}
\begin{split}
{\hat{C}}=1-\left(\frac{h}{\left(h+\frac{\vartheta\bar\gamma}{2\mu}\right)^2-H^2}\right)^{\mu}.
\end{split}
\end{equation} 
\end{lemma}

\begin{proof}
By Equ. \eqref{eq3} and Equ. \eqref{eq14}, the approximated EMI is derived as follows:
\begin{equation}
\label{eq16}
\begin{split}
C\approx&1-\int_{0}^{+\infty}\frac{2\sqrt{\pi}\mu^{\mu+\frac{1}{2}}h^\mu}{\Gamma\left(\mu\right)H^{\mu-\frac{1}{2}}}\frac{\gamma^{\mu-\frac{1}{2}}}{\bar\gamma^{\mu+\frac{1}{2}}}\exp\left(-\frac{2\mu\gamma h}{\bar\gamma}-\vartheta\gamma\right)\\
&\qquad\times I_{\mu-\frac{1}{2}}\left(\frac{2\mu H\gamma}{\bar\gamma}\right){\rm{d}}\gamma
\\
\overset{\left(a\right)}{=}&1-
\frac{2\sqrt{\pi}h^\mu}{\Gamma\left(\mu\right)\left(2h+\frac{\vartheta\bar\gamma}{\mu}\right)^{2\mu}}\sum_{n=0}^{+\infty}\frac{\Gamma\left(2\mu+2n\right)}{\Gamma\left(\mu+\frac{1}{2}+n\right)}\frac{1}{n!}\\
&\qquad\times\frac{1}{\left(\frac{2h}{H}+\frac{\bar\gamma\vartheta}{\mu H}\right)^{2n}}\\
\overset{\left(b\right)}{=}&1-\frac{h^{\mu}}{\left(h+\frac{\vartheta\bar\gamma}{2\mu}\right)^{2\mu}}\sum_{n=0}^{+\infty}\frac{\Gamma\left(\mu+n\right)}{n!\Gamma\left(\mu\right)}\frac{1}{\left(\frac{h}{H}+\frac{\vartheta\bar\gamma}{2\mu H}\right)^{2n}}\\
\overset{\left(c\right)}{=}&1-\frac{h^{\mu}}{\left(h+\frac{\vartheta\bar\gamma}{2\mu}\right)^{2\mu}}\left(1-\frac{1}{\left(\frac{h}{H}+\frac{\vartheta\bar\gamma}{2\mu H}\right)^{2}}\right)^{-\mu}\\
=&1-\left(\frac{h}{\left(h+\frac{\vartheta\bar\gamma}{2\mu}\right)^2-H^2}\right)^{\mu},
\end{split}
\end{equation} 
where the step ``(a)'' is due to \cite[Equ. (3.326.2)]{b2}, the step ``(b)'' is due to the Legendre duplication formula
\begin{equation}
\label{eq17}
\Gamma\left(z\right)\Gamma\left(z+\frac{1}{2}\right)=2^{1-2z}\sqrt{\pi}\Gamma\left(2z\right),
\end{equation} 
and the step ``(c)'' is due to the following Taylor series 
\begin{equation}
\label{EQU18}
\left(1-x^2\right)^{-m}=\sum_{n=0}^{+\infty}\frac{\Gamma\left(n+m\right)}{\Gamma\left(m\right)n!}x^{2n},\qquad |x|<1.
\end{equation} 
Notably, for any format of $\eta$-$\mu$ fading, it is easy to prove $\left|\frac{h}{H}+\frac{\vartheta\bar\gamma}{2\mu H}\right|>1$ with the expressions listed on the last page, thus Equ. \eqref{EQU18} can be directly utilized.
\end{proof}
By \cite{b5}, the Nakagami-$m$ fading, given in Equ. \eqref{eq7}, can be obtained from anyone
of the two formats of $\eta$-$\mu$ fading: 1) for Format 1, $\mu=m$ and $\eta\rightarrow{0}$ or $\eta\rightarrow{\infty}$; 2) for Format 2, $\mu=m$ and $\eta\rightarrow{+1}$ or $\eta\rightarrow{-1}$. In both formats, there are $\left(h-H\right)\rightarrow0.5$, $\frac{H}{h}\rightarrow1$ and $h\rightarrow\infty$, thus
\begin{equation}
\begin{split}
{\hat{C}}=&1-\left(\frac{h}{\left(h+\frac{\vartheta\bar\gamma}{2\mu}\right)^2-H^2}\right)^{\mu}\\
=&1-\left(\frac{1}{\left(1+\frac{\vartheta\bar\gamma}{2\mu h}+\frac{H}{h}\right)\left(h+\frac{\vartheta\bar\gamma}{2\mu}-H\right)}\right)^{\mu}\\
=&1-\left(\frac{\mu}{\mu+\vartheta\bar\gamma}\right)^{\mu}=1-\left(\frac{m}{m+\vartheta\bar\gamma}\right)^{m},
\end{split}
\end{equation}
which is consistent with Equ. \eqref{eq8}.

\subsection{$\kappa$-$\mu$ Fading Scenario}
The PDF of received SNR in $\kappa$-$\mu$ fading channels is formulated as \cite{b5}
\begin{equation}
\label{eq18}
\begin{split}
f\left(\gamma\right)=&\frac{\mu\left(1+\kappa\right)^{\frac{\mu+1}{2}}}{\kappa^{\frac{\mu-1}{2}}\exp\left(\mu\kappa\right)}\frac{\gamma^{\frac{\mu-1}{2}}}{\bar\gamma^{\frac{\mu+1}{2}}}
\exp\left(-\frac{\mu\left(1+\kappa\right)\gamma}{\bar\gamma}\right)\\
&\times I_{\mu-1}\left(2\mu\sqrt{\frac{\kappa\left(\kappa+1\right)\gamma}{\bar\gamma}}\right).
\end{split}
\end{equation}
\begin{lemma}
The ergodic mutual information under $\kappa$-$\mu$ fading is approximated as
\begin{equation}
\label{eq19}
\begin{split}
{\hat{C}}=1-\frac{1}{\left(1+\frac{\vartheta\bar\gamma}{\mu\left(1+\kappa\right)}\right)^{\mu}}
\exp\left(\frac{\mu\kappa}{1+\frac{\vartheta\bar\gamma}{\mu\left(1+\kappa\right)}}-\mu\kappa\right).
\end{split}
\end{equation} 
\end{lemma}
\begin{proof}
On the basis of Equ. \eqref{eq3} and Equ. \eqref{eq18},
\begin{equation}
\label{eq20}
\begin{split}
C\approx&1-\mu^{\mu-1}\left(\frac{\kappa\left(\kappa+1\right)}{\bar\gamma}\right)^{\frac{\mu-1}{2}}\frac{\mu\left(1+\kappa\right)^{\frac{\mu+1}{2}}}{\kappa^{\frac{\mu-1}{2}}\exp\left(\mu\kappa\right)}\frac{1}{\bar\gamma^{\frac{\mu+1}{2}}}\\
&\times\int_{0}^{+\infty}\sum_{n=0}^{+\infty}\exp\left(-\left(\frac{\mu\left(1+\kappa\right)}{\bar\gamma}+\vartheta\right)\gamma\right)\\
&\times\mu^{2n}\frac{1}{n!\Gamma\left(\mu+n\right)}\left(\frac{\kappa\left(\kappa+1\right)}{\bar\gamma}\right)^n\gamma^{n+\frac{\mu-1}{2}+\frac{\mu-1}{2}}{\rm{d}}\gamma\\
=&1-\frac{\exp\left(-\mu\kappa\right)}{\left(1+\frac{\vartheta\bar\gamma}{\mu\left(1+\kappa\right)}\right)^{\mu}}\sum_{n=0}^{+\infty}
\frac{1}{n!}\left(\frac{\mu\kappa}{1+\frac{\vartheta\bar\gamma}{\mu\left(1+\kappa\right)}}\right)^n\\
=&1-\frac{1}{\left(1+\frac{\vartheta\bar\gamma}{\mu\left(1+\kappa\right)}\right)^{\mu}}
\exp\left(\frac{\mu\kappa}{1+\frac{\vartheta\bar\gamma}{\mu\left(1+\kappa\right)}}-\mu\kappa\right).
\end{split}
\end{equation} 
It should be noted that \cite[Equ. (3.326.2)]{b2} is utilized during the calculation of this integral.
\end{proof}

It is known that $\kappa$-$\mu$ fading comprises both Rician ($\mu=1$) and Nakagami-$m$ ($\kappa\rightarrow0$) \cite{b5}. When $\mu=1$, Equ. \eqref{eq20} turns into
\begin{equation}
\label{eq21}
\begin{split}
C\approx&1-\frac{1}{1+\frac{\vartheta{\bar\gamma}}{\kappa+1}}\exp\left(\frac{\kappa}{1+\frac{\vartheta{\bar\gamma}}{\kappa+1}}-\kappa\right),
\end{split}
\end{equation} 
which is consistent with Equ. \eqref{eq12}. In addition, Nakagami fading arises form Equ. \eqref{eq18} for $\kappa\rightarrow0$, thus its ergodic mutual information reads
\begin{equation}
\label{eq22}
\begin{split}
C\approx&1-\left(\frac{\mu}{\bar\gamma}\right)^\mu
\frac{1}{\left(\frac{\mu}{\bar\gamma}+\vartheta\right)^{\mu}}=1-\left(\frac{\mu}{\mu+\vartheta\bar\gamma}\right)^\mu,
\end{split}
\end{equation} 
which accords with Equ. \eqref{eq8}.

\section{Simulation}  
\label{section3}
In this part, numerically evaluated results are plotted and compared with Monte Carlo simulation 
to verify the accuracy of the former derivations.
Notably, Rayleigh, Rician, $\eta$-$\mu$ and $\kappa$-$\mu$ fading can all be generated from Gaussian randoms, and the Nakagami-$m$ fading can be directly obtained by the Matlab function $random()$.
\begin{figure}[!t] 
\setlength{\abovecaptionskip}{-4pt}  
\centering 
\includegraphics[width=0.45\textwidth]{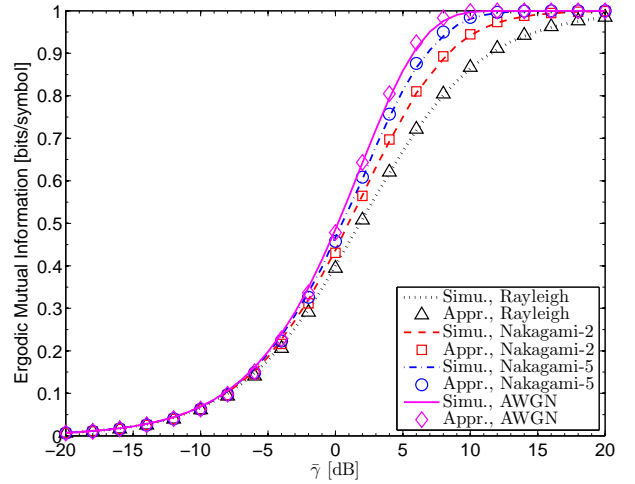} 
\caption{Approximated (Appr.) and simulated (Simu.) ergodic mutual information for Rayleigh and Nakagami-$m$ fadings.}
\label{figure1}
\vspace{-4pt}
\end{figure}

\begin{figure}[!t] 
\setlength{\abovecaptionskip}{-4pt}  
\centering 
\includegraphics[width=0.45\textwidth]{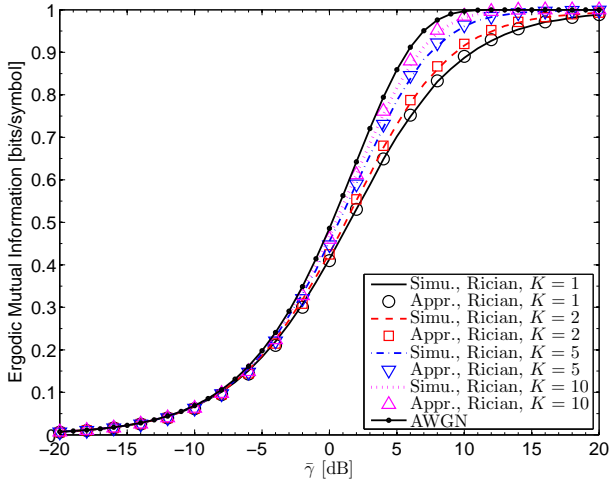} 
\caption{Approximated and simulated ergodic mutual information for Rician fading.}
\label{figure2}
\vspace{-4pt}
\end{figure}

\begin{figure*}[!t]   
    \centering
\setlength{\abovecaptionskip}{-4pt}
    \subfigure[Format 1]
    {
        \includegraphics[width=0.45\textwidth]{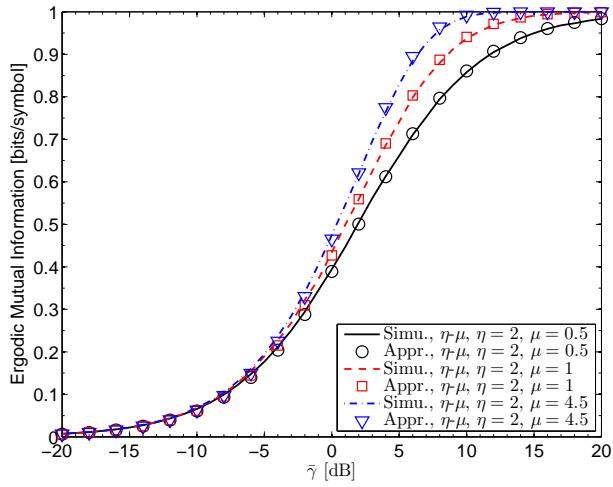}
	   \label{fig3a}	   
    } 
   \subfigure[Format 2]
    {
        \includegraphics[width=0.45\textwidth]{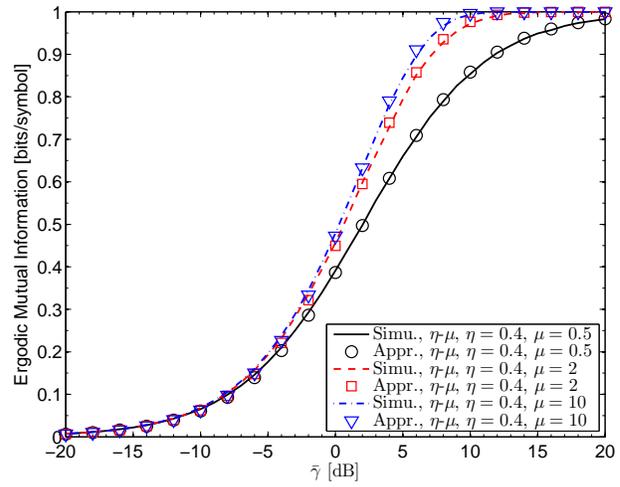}
	   \label{fig3b}	   
    } 
\\
    \caption{Approximated and simulated ergodic mutual information for $\eta$-$\mu$ fading. The results of Format 1 and Format 2 are illustrated in {\figurename} \ref{fig3a} and {\figurename} \ref{fig3b}, respectively.}
    \label{figure3}
	\vspace{-4pt}
\end{figure*}

\begin{figure}[!t] 
\setlength{\abovecaptionskip}{-4pt}  
\centering 
\includegraphics[width=0.45\textwidth]{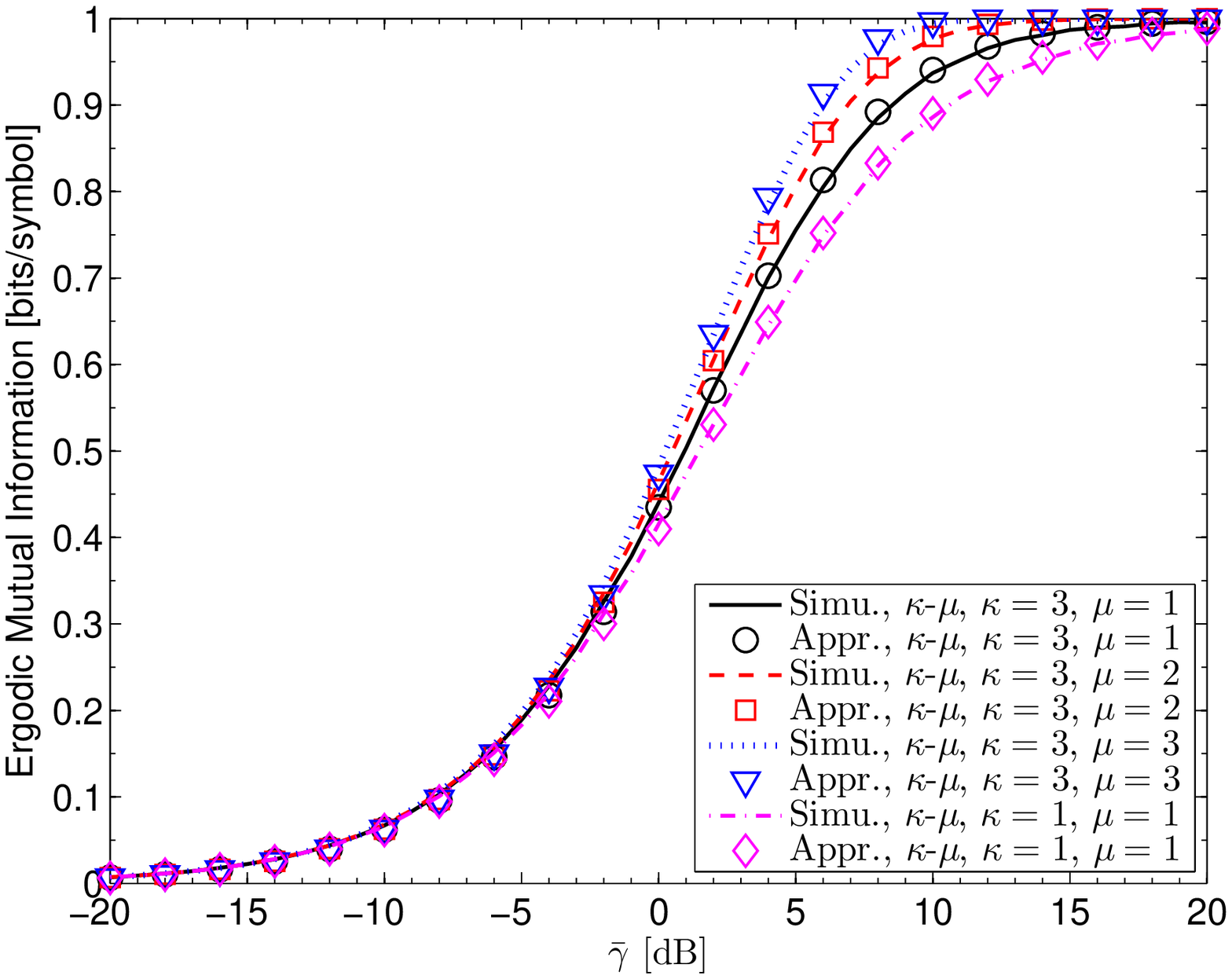} 
\caption{Approximated and simulated ergodic mutual information for $\kappa$-$\mu$ fading.}
\label{figure4}
\vspace{-4pt}
\end{figure}

{\figurename} \ref{figure1} compares the approximated and simulated ergodic mutual information in terms of $\bar\gamma$ for AWGN, Rayleigh and Nakagammi-$m$ channels. Firstly, let us focus on the scenario of AWGN, whose approximated values are calculated by Equ. \eqref{eq3}. As it shows, the simulated results meet accurately with the derivations, which suggests that Equ. \eqref{eq3} possesses high approximation precision and it is accurate enough to apply Equ. \eqref{eq3} into the estimation of mutual information. Most importantly, it can be observed that the approximation matches well with the empirical distribution for both Rayleigh and Nakagami-$m$ fading. Besides, it can be seen from this figure that the Nakagami-$m$ fading channel tends to be AWGN as $m$ grows. Then, {\figurename} \ref{figure2} shows the ergodic mutual information as a function of $\bar\gamma$ under Rician fading for selected values of $K$. It can be seen that the derivations meet tightly the results given via numerical simulations for all $K$ and SNR. Moreover, with the increment of $K$, the negative effect from multi-path fading degrades, which contributes to the convergence trend toward AWGN in the graph above. Later, {\figurename} \ref{figure3} and {\figurename} \ref{figure4} move on to illustrate more generalized fading types, i.e. $\eta$-$\mu$ and $\kappa$-$\mu$ fading. In {\figurename} \ref{figure3}, both Format 1 and Format 2 are presented to validate Equ. \eqref{eq15}. As can be seen from both {\figurename} \ref{figure3} and {\figurename} \ref{figure4}, the proposed approximation achieves a good agreement with the simulation. 

\section{Conclusion}
\label{section4}
This paper gives some closed-form approximate formulas of the ergodic mutual information in generalized fading channels under BPSK modulation. 
For each fading scenario, our derivations can meet tightly the empirical results. Numerical experiments suggest that our approximation provides a simple but numerically efficient way to calculate the mutual information of generalized fading channels.

\vspace{12pt}
\end{document}